\documentclass[11pt]{amsart}
\usepackage{amsmath,amssymb,amsbsy,amsfonts,amsthm,latexsym,
            amsopn,amstext,amsxtra,euscript,amscd,color,mathrsfs}

\PassOptionsToPackage{hyphens}{url}\usepackage{hyperref}
 
 \usepackage[capbesideposition=outside,capbesidesep=quad]{floatrow}

\restylefloat{table}
\restylefloat{table}
            
\usepackage{multirow,caption}
            
\usepackage{amscd}
\usepackage{color,enumerate}

\newcommand{\RNum}[1]{\lowercase\expandafter{\romannumeral #1\relax}}

\usepackage[colorinlistoftodos,prependcaption,textsize=tiny]{todonotes}

\newtheorem{thm}{Theorem}[section]
\newtheorem{lem}[thm]{Lemma}

\newtheorem{thm-con}[thm]{Theorem-Conjecture}
\numberwithin{equation}{section}

\theoremstyle{definition}

\def\cB{{\mathcal B}}
\def\cC{{\mathcal C}}

\newcommand{\F}{\mathbb{F}} 
\def\Tr{{\rm Tr}}

\begin{document}
\title[ Some new classes of (almost) perfect $c$-nonlinear permutations]{ Some new classes of (almost) perfect $c$-nonlinear permutations}

\author[M. Pal]{Mohit Pal}
\address{Department of Informatics, University of Bergen, PB 7803, 5020, Bergen, Norway}
\email{mohit.pal@uib.no}

\begin{abstract}
The concept of differential uniformity was recently extended to the $c$-differential uniformity. An interesting problem in this area is the construction of functions with low $c$-differential uniformity and a lot of research has been done in this direction in the recent past. Here, we present three classes of  (almost) perfect $c$-nonlinear permutations over finite fields of even characteristic.
\end{abstract}

\maketitle 

\section{Introduction}
Let $\F_{q^n}$ be the finite field with $q^n$ elements, where $q$ is a prime power and $n$ is a positive integer. We denote by $\F_{q^n}^*$, the multiplicative cyclic group of non-zero elements of $\F_{q^n}$. Let $f$ be a function from $\F_{q^n}$ to itself. Lagrange's interpolation formula allows us to write $f$ as a polynomial in $\F_{q^n}[X]/ (X^{q^n}-X)$ uniquely. Thus, we shall use the term function and polynomial for $f$ interchangeably. A polynomial $f\in \F_{q^n}[X]$ is called a permutation polynomial if the associated mapping $c \mapsto f(c)$ permutes the elements of $\F_{q^n}$. Here we consider permutation polynomials of the shape $f(X) = \gamma X+ \Tr_{q^n/q}(X^k)$, where $\Tr_{q^n/q} : \F_{q^n} \rightarrow \F_q$ is the trace map given by $\Tr_{q^n/q}(X)= \sum_{i=0}^{n-1}X^{q^i}$ and $\gamma \in \F_{q^n}^*$. These permutation polynomials are interesting because of their simple algebraic structure. Charpin and Kyureghyan~\cite{CK2008} first studied the polynomials of the form $f$ and gave a complete characterization of permutation polynomials for $q=2$. In 2016, Kyureghyan and Zieve~\cite{KZ2016} studied permutation polynomials of the shape $X+\gamma \Tr_{q^n/q}(X^k)$ over $\F_{q^n}$ for odd $q$ and constructed several classes of permutation polynomials. Motivated by the results of~\cite{CK2008, KZ2016}, Li et al.~\cite{LQCL2018} studied polynomials of the form $\gamma X+ \Tr_{q^n/q}(X^k)$ over finite fields of even characteristic. To avoid repetitive work from~\cite{CK2008}, the authors in~\cite{LQCL2018} restricted themselves to the case $q>2$ and constructed fifteen new classes of permutation polynomials. For more details on permutation polynomials of the shape  $\gamma X+ \Tr_{q^n/q}(X^k)$, reader may refer to ~\cite{CK2008, GK2019, K2011, KZ2016, LQCL2018, MG2017} and references therein.

Permutation polynomials have wide applications in the areas of mathematics and engineering such as coding theory~\cite{DH13, LC07}, cryptography~\cite{LM84, SH98} and combinatorial designs~\cite{DY06}. In symmetric key cryptography, permutation polynomials are widely applied to the design of block ciphers, where they are used in the design of so-called substitution boxes. In fact, permutation polynomials must be choosen in a way so as S-box becomes immune to some of the known attacks against block ciphers. For instance, the differential attack introduced by Biham and Shamir~\cite{BS1991} is one of the most powerful attacks on block ciphers. In 1993, Nyberg~\cite{N1993} introduced the notion of differential uniformity to quantify the immunity of a function against differential attacks. Recently, Ellingsen et al.~\cite{EFRST2020} generalised the notion of differential uniformity and defined the $c$-differential uniformity of a function $f$ for any $c\in \F_{q^n}$. For any $a,c \in \F_{q^n}$, the $c$-derivative of $f$ in the direction of $a$ is given by $_cD_f(a,X):= f(X+a)-cf(X)$ for all $X \in \F_{q^n}$. For any $a,b \in \F_{q^n}$, the $c$-Difference Distribution Table ($c$-DDT) entry at point $(a,b)$ of $f$ is given by $ _c\Delta_f(a,b) := \lvert \{ X \in \F_{q^n} \mid {_cD}_f(a,X) =b \} \rvert$. The $c$-differential uniformity of $f$ is defined as $_c\Delta_f := \max \{ _c\Delta_f(a,b) \mid a, b \in \F_{q^n}~\mbox{and}~a \neq 0~\mbox{when}~c=1 \}$. When $_c\Delta_f=1,2,$ then the function $f$ is called perfect $c$-nonlinear (P$c$N) function and almost perfect $c$-nonlinear (AP$c$N) function, respectively. When $c=1$, then the notions of $c$-DDT and $c$-differential uniformity coincide with the classical notions of DDT and differential uniformity, respectively.

Since the appearence of the notion of the $c$-differential uniformity, several functions having low $c$-differential uniformity have been constructed in the recent past. The most recent list of power maps along with their $c$-differential uniformity can be found in~\cite[Table 1]{LRS2022}, and the list of non-power maps with their $c$-differential uniformity is given in~\cite[Table 2]{LRS2022} for $c \neq 1$. However, over finite fields of even characteristic, there are only few P$c$N and AP$c$N functions, see, e.g.,~\cite{EFRST2020,HPS22,JKK2022, MRSYZ21,TZJT}. In this paper, we shall restrict ourselves to finite fields of even characteristic and study the $c$-differential uniformity of three classes of permutation polynomials of the shape $\gamma X+ \Tr_{q^n/q}(X^k)$ over $\F_{q^n}$ for $c \in \F_{q^n} \backslash \{1\}$. The rationale behind excluding the case $c=1$ is the fact that in this case linear term outside the trace function vanishes and the varible $X$ remains only in the trace term. As a consequence, polynomials of the form $\gamma X+ \Tr_{q^n/q}(X^k)$ will have very high differential uniformity.

The rest of the paper is organised as follows. In Section~\ref{S2}, we consider the $c$-differential uniformity of the permutation polynomial $f(X):= \gamma X +\Tr_{q^{2n}/q} (X^{2^i(q+1)})$, where $\gamma \in  \F_{q^2}^*$, $q=2^t$ and $i,n, t >0$ are integers. We show that $f$ is P$c$N for all $c\in \F_{q^2} \backslash \F_2$ and if $\gcd(i,t)=1$ then $_c\Delta_f \leq 2$ for all $c\in \F_{q^{2n}} \backslash \F_{q^2}$. Since we are considering only permutation polynomials, and for $c=0$ every permutation polynomial is P$c$N, we have excluded the trivial case $c=0$ in rest of the paper. In Section~\ref{S3}, we shall show that the permutation polynomial $g(X):= \gamma X +\Tr_{q^{2n}/q} (X^{q^2+1})$, where $\gamma \in  \F_{q}^*$, $q=2^t$ and $n, t >0$ are integers, is P$c$N for all $c\in \F_{q^{2n}} \backslash \F_2$. Section~\ref{S4} is devoted to the $c$-differential uniformity of the permutation polynomial $h(X):= \gamma X +\Tr_{q^{2n}/q} (X^{2^i(q^2+1)})$, where $\gamma \in  \F_{q^2}^*$, $q=2^t$ and $i,n, t >0$ are integers and either $n$ is even; or $n$ is odd and $\displaystyle \Tr_{q^2/q} \left( \left( \frac{1}{\gamma} \right)^{2^{i+1}} \right)^{\frac{q-1}{\gcd(2^{i+1}-1, 2^t-1)}} \neq 1$. We show that $h$ is P$c$N for all $c\in \F_{q^2} \backslash \F_2$, and if $\gcd(i,t)=1$ then the $c$-differential uniformity of $h$ is $\leq 2$ for all $c\in \F_{q^{2n}} \backslash \F_{q^2}$. Finally, we summarize the paper with some concluding remarks in Section~\ref{S5}.

\section{The $c$-differential uniformity of the function $f$}\label{S2}
In this section, we shall consider the $c$-differential uniformity of the function $f(X)= \gamma X +\Tr_{q^{2n}/q} (X^{2^i(q+1)})$, where $\gamma \in  \F_{q^2}^*$, $q=2^t$ and $i,n, t >0$ are integers. It has been proved in~\cite[Theorem 3.1]{LQCL2018} that $f$ is a permutation polynomial over $\F_{q^{2n}}$. We first recall a lemma that will be used throughout this paper.
\begin{lem}\cite{JL1978}
\label{L1}
 Let $q=p^t$, $i$ be a positive integer and $f(X)= X^{p^i}-\alpha X-\beta$ where $\alpha \in \F_q^*$ and $\beta \in \F_q$. Then, in the field $\F_q$, $f$ has either zero, one or $p^d$ roots where $d=\gcd(i,t)$.
\end{lem}
The following theorem gives the $c$-differential uniformity of the function $f$ for all $c \in \F_{q^{2n}} \backslash \F_2$.

\begin{thm}\label{T1}
Let $q=2^t$ and $f(X)= \gamma X +\Tr_{q^{2n}/q} (X^{2^i(q+1)})$, where $\gamma \in  \F_{q^2}^*$ and $i,n, t >0$ are integers. Then
\begin{enumerate}[(i)]
 \item $f$ is P$c$N over $\F_{q^{2n}}$ for all $c\in \F_{q^2} \backslash \F_2$,
 \item $_c\Delta_f \leq 2$ for all $c\in \F_{q^{2n}} \backslash \F_{q^2}$ when $\gcd(i,t)=1$.
\end{enumerate}

\end{thm}
\begin{proof}
Let $c \in \F_{q^{2n}}\backslash \F_2 $. Recall that, for any $(a,b) \in \F_{q^{2n}} \times \F_{q^{2n}}$, the $c$-DDT entry $_c \Delta_f (a,b)$ is given by the number of solutions $X \in \F_{q^{2n}}$ of the following equation.
\begin{equation} \label{T1E1}
 f(X+a)+cf(X)= b.
\end{equation} 
It is easy to see that when $a=0$ then for each $b \in \F_{q^{2n}}$, Equation~\eqref{T1E1} has a unique solution $X \in \F_{q^{2n}}$. For $a \neq 0$, Equation~\eqref{T1E1} can be written as
\begin{equation}\label{T1E2}
\begin{split}
 \gamma X +\Tr_{q^{2n}/q} (X^{2^i(q+1)}) + C\cdot  \Tr_{q^{2n}/q}((X^qa+a^qX)^{2^i})&= B,
\end{split}
\end{equation}
where $\displaystyle B = C ( b + \gamma a +\Tr_{q^{2n}/q} (a^{2^i(q+1)})) $ and $\displaystyle C = \frac{1}{1+c} $. Notice that, when $a \in \F_q^*$, then the above equation reduces to 
\begin{align*}
B &= \gamma X +\Tr_{q^{2n}/q} (X^{2^i(q+1)}) + a^{2^i}C \cdot \Tr_{q^{2n}/q}((X^{2^i})^q+X^{2^i})= \gamma X +\Tr_{q^{2n}/q} (X^{2^i(q+1)}),
\end{align*}
which has a unique solution for each $B \in \F_{q^{2n}}$ (and hence for each $b \in \F_{q^{2n}}$). Now, let $a \in \F_{q^{2n}} \backslash \F_q$. Then, from Equation~\eqref{T1E2} we have
\begin{equation}\label{T1E3}
 \gamma X +\Tr_{q^{2n}/q} (X^{2^i(q+1)}) + C \cdot \Tr_{q^{2n}/q}\left(((X^{q^2}+X)a^q)^{2^i}\right)= B.
\end{equation}
We shall now study the solutions of Equation~\eqref{T1E3}, in three different cases, namely, $c \in \F_q \backslash \F_2$, $c \in \F_{q^2} \backslash \F_q $ and $c \in \F_{q^{2n}} \backslash \F_{q^2}$, respectively.

\noindent
\textbf{Case 1} Let $c \in \F_q \backslash \F_2$. Now, let $\displaystyle u := \Tr_{q^{2n}/q} (X^{2^i(q+1)}) + C \cdot \Tr_{q^{2n}/q}\left(((X^{q^2}+X)a^q)^{2^i}\right)$. Then, from Equation~\eqref{T1E3}, we have $\displaystyle X = \frac{u+B}{\gamma}$. Putting the value of $X$ into Equation~\eqref{T1E3}, we have
\begin{equation*}
\begin{split}
 u &=  \frac{1}{\gamma^{2^i(q+1)}}\Tr_{q^{2n}/q} \left( \left( u^2+(B^q+B)u+B^{q+1} \right)^{2^i}\right) + CB_1\\
 &=  \frac{1}{\gamma^{2^i(q+1)}}\Tr_{q^{2n}/q} (B^{2^i(q+1)}) + CB_1,\\
 \end{split}
\end{equation*}
where $\displaystyle B_1 = \Tr_{q^{2n}/q}\left(\left(\frac{(B^{q^2}+B)a^q}{\gamma}\right)^{2^i}\right) $ and the first equality holds as $u, \gamma^{q+1} \in \F_q$ and the second equality holds as $\Tr_{q^{2n}/q} (1)=0$ and $\Tr_{q^{2n}/q} (B^q+B)=0$. Thus, for $c \in \F_q \backslash \F_2$, Equation~\eqref{T1E3} has a unique solution for each $a\in \F_{q^{2n}} \backslash \F_q$ and $b \in \F_{q^{2n}}$. 

\noindent
\textbf{Case 2} Let $c \in \F_{q^2} \backslash \F_q $. Let $\displaystyle u:= \Tr_{q^{2n}/q} (X^{2^i(q+1)})$ and $v:= \Tr_{q^{2n}/q}\left(((X^{q^2}+X)a^q)^{2^i}\right)$. Then, from Equation~\eqref{T1E3}, we have $\displaystyle X = \frac{u+Cv+B}{\gamma}$. Using the properties of the trace function, we shall now simplify $u$ and $v$. Consider 
\begin{equation*}
\begin{split}
 u &= \frac{1}{\gamma^{2^i(q+1)}}\Tr_{q^{2n}/q} \left(\left(u^2+((B+Cv)^q+B+Cv)u+ C^{q+1}v^2+ (B^qC+BC^q)v +B^{q+1}\right)^{2^i}\right) \\
 &= \frac{1}{\gamma^{2^i(q+1)}}\Tr_{q^{2n}/q} \left(\left(((BC^q)^q+BC^q)v +B^{q+1}\right)^{2^i}\right) \\
 &= \frac{1}{\gamma^{2^i(q+1)}}\Tr_{q^{2n}/q} \left( B^{2^i(q+1)} \right),
 \end{split}
\end{equation*}
where the first equality holds because $u,v \in \F_q$ and $C \in \F_{q^2}$ and the second equality holds because $ C^{q+1} \in \F_q$, $\Tr_{q^{2n}/q} ((B+Cv)^{q\cdot 2^i})= \Tr_{q^{2n}/q}((B+Cv)^{2^i})$ and the third equality holds as $\Tr_{q^{2n}/q}((BC^q)^{q\cdot 2^i})=\Tr_{q^{2n}/q}((BC^q)^{2^i})$. Now consider $v$ which is given by
\begin{equation*}
\begin{split}
v &= \Tr_{q^{2n}/q}\left(\left(\left(\frac{u+Cv+B}{\gamma}\right)^{q^2}+\left(\frac{u+Cv+B}{\gamma}\right)\right)^{2^i}a^{q\cdot2^i}\right)=B_1,
 \end{split}
\end{equation*}
where the last equality holds as $u, v \in \F_q$ and $C, \gamma \in \F_{q^2}$. Thus, $c \in \F_{q^2} \backslash \F_q $, Equation~\eqref{T1E3} has a unique solution for each $a\in \F_{q^{2n}} \backslash \F_q$ and $b \in \F_{q^{2n}}$. 

\noindent
\textbf{Case 3} Let $c \in \F_{q^{2n}} \backslash \F_{q^2}$. Again, let $\displaystyle u:= \Tr_{q^{2n}/q} (X^{2^i(q+1)})$, $v:= \Tr_{q^{2n}/q}\left(((X^{q^2}+X)a^q)^{2^i}\right)$. Then, from Equation~\eqref{T1E3}, we have
\begin{equation}\label{T1E4}
 X = \frac{u+Cv+B}{\gamma}.
\end{equation}
We shall now simplify $u$ and $v$. Consider $u$ which is given by
\begin{equation*}
\begin{split}
 u =&~ \frac{1}{\gamma^{2^i(q+1)}}\Tr_{q^{2n}/q} \left((u^2+((B+Cv)^q+B+Cv)u+ C^{q+1}v^2+ (BC^q+B^qC)v +B^{q+1})^{2^i}\right) \\
=& \frac{1}{\gamma^{2^i(q+1)}} \left( v^{2^{i+1}} \cC + v^{2^i} \cB + \Tr_{q^{2n}/q}( B^{2^i(q+1)}) \right),\\
\end{split}
\end{equation*}
where $\displaystyle \cB = \Tr_{q^{2n}/q}((BC^q+B^qC)^{2^i})$, $\displaystyle \cC = \Tr_{q^{2n}/q} (C^{2^i(q+1)})$ and the first equality holds because $u,v \in \F_q$, and last equality holds as $\Tr_{q^{2n}/q} ((B+Cv)^{q\cdot 2^i})= \Tr_{q^{2n}/q}((B+Cv)^{2^i})$. Now consider $v$ which is given by 
\begin{equation*}
\begin{split}
v = \Tr_{q^{2n}/q}\left(\left(\left(\frac{u+Cv+B}{\gamma}\right)^{q^2}+\left(\frac{u+Cv+B}{\gamma}\right)\right)^{2^i}a^{q\cdot2^i}\right) = v^{2^i} C_1+B_1,
 \end{split}
\end{equation*}
where $\displaystyle C_1= \Tr_{q^{2n}/q}\left(\left(\frac{(C^{q^2}+C)a^q}{\gamma}\right)^{2^i}\right) $  and the last equality holds as $u, v \in \F_q$ and $ \gamma \in \F_{q^2}$. When $C_1=0$, then $v=B_1$ and by putting the values of $u$ and $v$ in Equation~\eqref{T1E4}, we get a unique solution
\begin{equation*}
  X = \frac{1}{\gamma^{2^i(q+1)+1}} \left( (B_1)^{2^{i+1}} \cC + (B_1)^{2^i} \cB + \Tr_{q^{2n}/q}( B^{2^i(q+1)}) \right) + \frac{(CB_1+B)}{\gamma}
\end{equation*}
of Equation~\eqref{T1E3}. In the event when $C_1 \neq 0$, consider the equation
\[
 v^{2^i}+ C_1^{-1}v+B_1 C_1^{-1} = 0.
\]
Since $\gcd(i,t)=1$, from Lemma~\ref{L1} the above equation can have at most two solutions $v \in \F_q$. Thus, for $c \in \F_{q^{2n}} \backslash \F_{q^2}$, Equation~\eqref{T1E3} can have at most two solutions for each $a\in \F_{q^{2n}} \backslash \F_q$ and $b \in \F_{q^{2n}}$. This completes the proof.
\end{proof}

\section{The $c$-differential uniformity of the function $g$}\label{S3}
In this section, we shall consider the $c$-differential uniformity of the function $g(X):= \gamma X +\Tr_{q^{2n}/q} (X^{q^2+1})$, where $\gamma \in  \F_{q}^*$, $q=2^t$ and $n, t >0$ are integers. In~\cite[Theorem 3.2]{LQCL2018}, Li et al. proved that $g$ is a permutation polynomial over $\F_{q^{2n}}$. The following theorem gives the $c$-differential uniformity of the function $g$ for all $c \in \F_{q^{2n}} \backslash \F_2$.

\begin{thm}\label{T2}
Let $q=2^t$ and $g(X)= \gamma X +\Tr_{q^{2n}/q} (X^{q^2+1})$, where $\gamma \in  \F_{q}^*$ and $n, t >0$ are integers. Then $g$ is P$c$N over $\F_{q^{2n}}$ for all $c\in \F_{q^{2n}} \backslash \F_2$
\end{thm}
\begin{proof}
 Let $c \in \F_{q^{2n}} \backslash \F_{2}$. For any $(a,b) \in \F_{q^{2n}} \times \F_{q^{2n}}$ the $c$-DDT entry $_c \Delta_g (a,b)$ is given by the number of solutions $X \in \F_{q^{2n}}$ of the following equation.
\begin{equation}\label{T2E1}
\begin{split}
 \gamma X +\Tr_{q^{2n}/q} (X^{q^2+1}) + C\cdot \Tr_{q^{2n}/q}((X^{q^2}a+a^{q^2}X))&= B,
\end{split}
\end{equation}
where $\displaystyle B = C\left(b + \gamma a +\Tr_{q^{2n}/q} (a^{q^2+1})\right)$ and $\displaystyle C =  \frac{1}{1+c}$. When $a \in \F_{q}$, then the above equation reduces to $\gamma X +\Tr_{q^{2n}/q} (X^{q^2+1})= B$, which has a unique solution for each $B \in \F_{q^{2n}}$. Now, let $a \in \F_{q^{2n}} \backslash \F_q$ then Equation~\eqref{T2E1} can be rewritten as
\begin{equation}\label{T2E2}
 \gamma X +\Tr_{q^{2n}/q} (X^{q^2+1}) + C\cdot \Tr_{q^{2n}/q}((X^{q^4}+X)a^{q^2})= B.
\end{equation}
We shall split the analysis of the solutions of Equation~\eqref{T2E2} into three different cases, namely, $c \in \F_q \backslash \F_2$, $c \in \F_{q^2} \backslash \F_q$ and $c \in \F_{q^{2n}} \backslash \F_{q^2}$, respectively.

\noindent
\textbf{Case 1} Let $c \in \F_q \backslash \F_2$. Now, let $\displaystyle u:= \Tr_{q^{2n}/q} (X^{q^2+1}) + C \cdot \Tr_{q^{2n}/q}((X^{q^4}+X)a^{q^2})$. Then, from Equation~\eqref{T2E2}, we have $\displaystyle X = \frac{u+B}{\gamma}$. Putting the value of $X$ into Equation~\eqref{T2E2}, we have
\begin{equation*}
\begin{split}
 u &= \frac{1}{\gamma^2}\Tr_{q^{2n}/q} (u^{2}+(B^{q^2}+B)u+B^{q^2+1})
 + C\cdot \Tr_{q^{2n}/q}\left(\left(\frac{B^{q^4}+B}{\gamma}\right)a^{q^2}\right)\\
 &= \frac{1}{\gamma^2}\Tr_{q^{2n}/q} (B^{q^2+1})+CB_1,\\
 \end{split}
\end{equation*}
where $\displaystyle B_1 := \frac{1}{\gamma} \Tr_{q^{2n}/q}\left((B^{q^4}+B)a^{q^2}\right)$ and equalities hold because $u, \gamma \in \F_q$. Thus, for $c \in \F_{q} \backslash \F_2$, Equation~\eqref{T2E2} has a unique solution for each $a\in \F_{q^{2n}} \backslash \F_q$ and $b \in \F_{q^{2n}}$.

\noindent
\textbf{Case 2} Let $c \in \F_{q^2} \backslash \F_q$. In this case,  let $\displaystyle u:= \Tr_{q^{2n}/q} (X^{(q^2+1)})$ and  $v:= \Tr_{q^{2n}/q}((X^{q^4}+X)a^{q^2})$. Then, from Equation~\eqref{T2E2}, we have 
$\displaystyle X = \frac{u+Cv+B}{\gamma}$. Using properties of the trace function, we shall now simplify $u$ and $v$. Consider $u$ which is given by
\begin{equation*}
\begin{split}
 u =&~ \frac{1}{\gamma^2}\Tr_{q^{2n}/q} (u^2+(Cv)^2+(B^{q^2}+B)u+((BC)^{q^2}+BC)v +B^{q^2+1}) \\
 =& \frac{1}{\gamma^2} \left( v^2 \Tr_{q^{2n}/q} (C^2) + \Tr_{q^{2n}/q}( B^{q^2+1}) \right), \\
 \end{split}
\end{equation*}
where equalities hold because $u, v, \gamma \in \F_q$, $\Tr_{q^{2n}/q} (B^{q^2}+B)=0$ and $\Tr_{q^{2n}/q}((BC)^{q^2}+BC)=0$. Now consider $v$ which is given by
\begin{equation*}
\begin{split}
v &= \Tr_{q^{2n}/q}\left(\left(\left(\frac{u+Cv+B}{\gamma}\right)^{q^4}+\left(\frac{u+Cv+B}{\gamma}\right)\right)a^{q^2}\right) =B_1
\end{split}
\end{equation*}
where last equality holds as $u, v, \gamma \in \F_{q}$ and $C \in \F_{q^2}$. Thus, for $c \in \F_{q^2} \backslash \F_q$, Equation~\eqref{T2E2} has a unique solution for each $a\in \F_{q^{2n}} \backslash \F_q$ and $b \in \F_{q^{2n}}$.

\noindent 
\textbf{Case 3} Let $c \in \F_{q^{2n}} \backslash \F_{q^2}$. Again, let $\displaystyle u:= \Tr_{q^{2n}/q} (X^{(q^2+1)})$ and  $v:= \Tr_{q^{2n}/q}((X^{q^4}+X)a^{q^2})$. Then, from Equation~\eqref{T2E2}, we have
\begin{equation}\label{T2E3}
 X = \frac{u+Cv+B}{\gamma}.
\end{equation}
Now, consider $u$ which is given by 
\begin{equation*}
\begin{split}
  u =&~ \frac{1}{\gamma^2}\Tr_{q^{2n}/q} \left(u^2+((B+Cv)^{q^2}+B+Cv)u+ C^{q^2+1}v^2+ (BC^{q^2}+B^{q^2}C)v +B^{q^2+1}\right) \\
  =& \frac{1}{\gamma^2} \left( v^2 \cC + v\cB + \Tr_{q^{2n}/q}( B^{q^2+1}) \right),\\
\end{split}
\end{equation*}
where $\displaystyle \cB = \Tr_{q^{2n}/q}(BC^{q^2}+B^{q^2}C)$, $\displaystyle \cC = \Tr_{q^{2n}/q} (C^{q^2+1})$ and equalities hold because $u,v \in \F_q$ and $\Tr_{q^{2n}/q} ((B+Cv)^{q^2}+B+Cv)=0$. Now consider $v$ which is given by 
\begin{equation*}
\begin{split}
v&= \Tr_{q^{2n}/q}\left(\left(\left(\frac{u+Cv+B}{\gamma}\right)^{q^4}+\left(\frac{u+Cv+B}{\gamma}\right)\right)a^{q^2}\right) = vC_1+B_1,
\end{split}
\end{equation*}
where $\displaystyle C_1= \frac{1}{\gamma} \Tr_{q^{2n}/q}\left((C^{q^4}+C)a^{q^2}\right) $  and the last equality holds as $u, v, \gamma \in \F_{q}$. From above equation we get a unique solution for $v$. Thus, for $c \in \F_{q^{2n}} \backslash \F_{q^2}$, Equation~\eqref{T2E2} has a unique solution for each $a\in \F_{q^{2n}} \backslash \F_q$ and $b \in \F_{q^{2n}}$. This completes the proof.
\end{proof}

\section{The $c$-differential uniformity of the function $h$}\label{S4}

In this section, we shall consider the $c$-differential uniformity of a more general class of permutations than $g$. Let $h(X)= \gamma X +\Tr_{q^{2n}/q} (X^{2^i(q^2+1)})$, where $\gamma \in  \F_{q^2}^*$, $q=2^t$ and $i,n, t >0$ are integers. It has been proved in~\cite[Theorem 3.3]{LQCL2018} that $h$ is a permutation polynomial over $\F_{q^{2n}}$ if either $n$ is even; or $n$ is odd and $\displaystyle \Tr_{q^2/q} \left( \frac{1}{\gamma^{2^i+1}} \right)^{\frac{q-1}{\gcd(2^{i+1}-1, 2^t-1)}} \neq 1$. The following theorem gives the $c$-differential uniformity of the function $h$ for all $c \in \F_{q^{2n}} \backslash \F_2$.

\begin{thm}\label{T3}
Let $q=2^t$ and $h(X)= \gamma X +\Tr_{q^{2n}/q} (X^{2^i(q^2+1)})$, where $\gamma \in  \F_{q^2}^*$, $i,n, t >0$ are integers and either $n$ is even; or $n$ is odd and $\displaystyle \Tr_{q^2/q} \left( \frac{1}{\gamma^{2^i+1}} \right)^{\frac{q-1}{\gcd(2^{i+1}-1, 2^t-1)}} \neq 1$. Then 
\begin{enumerate}[(i)]
 \item $h$ is P$c$N over $\F_{q^{2n}}$ for all $c\in \F_{q^2} \backslash \F_2$,
 \item $_c\Delta_h \leq 2$ for all $c\in \F_{q^{2n}} \backslash \F_{q^2}$ when $\gcd(i,t)=1$.
\end{enumerate}
\end{thm}
\begin{proof}
 Let $c \in \F_{q^{2n}} \backslash \F_{2}$. For any $(a,b) \in \F_{q^{2n}} \times \F_{q^{2n}}$ the $c$-DDT entry $_c \Delta_h (a,b)$ is given by the number of solutions $X \in \F_{q^{2n}}$ of the following equation.
 \begin{equation}\label{T3E1}
\begin{split}
\gamma X +\Tr_{q^{2n}/q} (X^{2^i(q^2+1)}) + C \cdot \Tr_{q^{2n}/q}((X^{q^2}a+a^{q^2}X)^{2^i})&= B,
\end{split}
\end{equation}
where $\displaystyle B = C(b + \gamma a +\Tr_{q^{2n}/q} (a^{2^i(q^2+1)}))$ and $\displaystyle C = \frac{1}{1+c}$. Notice that when $a \in \F_{q}$, Equation~\eqref{T3E1} reduces to $B = \gamma X +\Tr_{q^{2n}/q} (X^{2^i(q^2+1)})$,   which has a unique solution for each $B \in \F_{q^{2n}}$. Now, let $a \in \F_{q^{2n}} \backslash \F_q$. Then, from Equation~\eqref{T3E1} we have
\begin{equation}\label{T3E2}
 \gamma X +\Tr_{q^{2n}/q} (X^{2^i(q^2+1)}) + C \cdot \Tr_{q^{2n}/q}\left(((X^{q^4}+X)a^{q^2})^{2^i}\right)= B.
\end{equation}
Now we shall split the analysis of the solutions of Equation~\eqref{T3E2} into three different cases, namely, $c \in \F_q \backslash \F_2$, $c \in \F_{q^2} \backslash \F_q$ and $c \in \F_{q^{2n}} \backslash \F_{q^2}$, respectively.

\noindent
\textbf{Case 1} Let $c \in \F_q \backslash \F_2$ and $\displaystyle u:= \Tr_{q^{2n}/q} (X^{2^i(q^2+1)}) + C \cdot \Tr_{q^{2n}/q}\left(((X^{q^4}+X)a^{q^2})^{2^i}\right).$ Then, from Equation~\eqref{T3E2}, we have $\displaystyle X = \frac{u+B}{\gamma}$. Putting the value of $X$ into Equation~\eqref{T3E2}, we have
\begin{equation*}
\begin{split}
 u &=  \Tr_{q^{2n}/q} \left( \left(\frac{u^{2}+(B^{q^2}+B)u+B^{q^2+1}}{\gamma^2}\right)^{2^i}\right)  + C_1 \\
 &=  \Tr_{q^{2}/q} \left( \Tr_{q^{2n}/q^2} \left( \left(\frac{u^{2}+(B^{q^2}+B)u+B^{q^2+1}}{\gamma^2}\right)^{2^i}\right) \right) +  C_1 \\
 &=  \Tr_{q^{2}/q} \left(  \left(\frac{u}{\gamma}\right)^{2^{i+1}}\Tr_{q^{2n}/q^2} (1) \right)+\Tr_{q^{2n}/q} \left( \left(\frac{B}{\gamma}\right)^{2^i(q^2+1)}\right)  +  C_1 \\
 &=
 \begin{cases}
  B_1+ C_1~&~\mbox{when}~n~\mbox{is even} \\
  u^{2^{i+1}}T_1 +B_1 +  C_1~&~\mbox{when}~n~\mbox{is odd},
 \end{cases}
 \end{split}
\end{equation*}
where $\displaystyle B_1 :=  \Tr_{q^{2n}/q} \left( \left(\frac{B}{\gamma}\right)^{2^i(q^2+1)}\right)$, $\displaystyle C_1 := C \cdot \Tr_{q^{2n}/q}\left(\left(\frac{(B^{q^4}+B)a^{q^2}}{\gamma}\right)^{2^i}\right)$ and $\displaystyle T_1 :=\Tr_{q^{2}/q} \left(\frac{1}{\gamma^{2^{i+1}}} \right)$. Notice that when $n$ is odd, $\displaystyle T_1^{\frac{q-1}{\gcd(2^{i+1}-1, 2^t-1)}} \neq 1$ and hence we get only one solution $u \in \F_q$.  Thus, for $c \in \F_q \backslash \F_2$, Equation~\eqref{T3E2} has a unique solution for each $a\in \F_{q^{2n}} \backslash \F_q$ and $b \in \F_{q^{2n}}$.

\noindent
\textbf{Case 2} Let $c \in \F_{q^2} \backslash \F_q$. Let $\displaystyle u:= \Tr_{q^{2n}/q} \left(X^{2^i(q^2+1)}\right)$ and $v:= \Tr_{q^{2n}/q}\left(((X^{q^4}+X)a^{q^2})^{2^i}\right).$ Then, from Equation~\eqref{T3E2}, we have 
\begin{equation}\label{T3E3}
 X = \frac{u+Cv+B}{\gamma}.
\end{equation}
Using the properties of the trace function, we shall now simplify $u$ and $v$. Consider $u$ which is given by
\begin{equation*}
\begin{split}
 u &= \Tr_{q^{2n}/q} \left(\left(\frac{u^2+(B^{q^2}+B)u+ C^2v^2+ ((BC)^{q^2}+BC)v +B^{q^2+1}}{\gamma^2} \right)^{2^i}\right) \\
 &= \Tr_{q^{2}/q} \left( \frac{1}{\gamma^{2^{i+1}}} \Tr_{q^{2n}/q^2} \left(\left(u^2+(B^{q^2}+B)u+ C^2v^2+ ((BC)^{q^2}+BC)v +B^{q^2+1} \right)^{2^i}\right) \right) \\
 &= \Tr_{q^{2}/q} \left( \frac{1}{\gamma^{2^{i+1}}} \Tr_{q^{2n}/q^2} \left(\left(u^2+ C^2v^2+B^{q^2+1} \right)^{2^i}\right) \right) \\
 &= \Tr_{q^{2}/q} \left(\frac{u^{2^{i+1}}}{\gamma^{2^{i+1}}} \Tr_{q^{2n}/q^2} \left(1\right)+ \frac{v^{2^{i+1}}}{\gamma^{2^{i+1}}} \Tr_{q^{2n}/q^2} \left(C^{2^{i+1}}\right) +\frac{1}{\gamma^{2^{i+1}}} \Tr_{q^{2n}/q^2} \left(B^{2^i(q^2+1)}\right) \right) \\
 &=
 \begin{cases}
  v^{2^{i+1}}C_2+B_1~&~\mbox{when}~n~\mbox{is even} \\
  u^{2^{i+1}}T_1+v^{2^{i+1}}C_2+B_1~&~\mbox{when}~n~\mbox{is odd},
 \end{cases}
 \end{split}
\end{equation*}
where $\displaystyle C_2 := \Tr_{q^{2n}/q} \left(\left(\frac{ C}{\gamma} \right)^{2^{i+1}}\right)$. Again, since $\displaystyle T_1^{\frac{q-1}{\gcd(2^{i+1}-1, 2^t-1)}} \neq 1$ when $n$ is odd, for any fixed $v$, we get a unique solution $u \in \F_q$ of the above equation. Now consider $v$ which is given by
\begin{equation*}
\begin{split}
v &= \Tr_{q^{2n}/q}\left(\left(\left(\frac{u+Cv+B}{\gamma}\right)^{q^4}+\left(\frac{u+Cv+B}{\gamma}\right)\right)^{2^i}a^{q^22^i}\right)= \frac{C_1}{C}.
 \end{split}
\end{equation*}
Now putting the values of $u, v$ into Equation~\eqref{T3E3}, we get a unique solution $X \in \F_{q^{2n}}$ of Equation~\eqref{T3E2}. Thus, for $c \in \F_{q^2} \backslash \F_q$, Equation~\eqref{T3E2} has a unique solution for each $a\in \F_{q^{2n}} \backslash \F_q$ and $b \in \F_{q^{2n}}$.

\noindent 
\textbf{Case 3} Let $c \in \F_{q^{2n}} \backslash \F_{q^2}$. Again, in this case too, let $\displaystyle u:= \Tr_{q^{2n}/q} \left(X^{2^i(q^2+1)}\right)$ and $v:= \Tr_{q^{2n}/q}\left(((X^{q^4}+X)a^{q^2})^{2^i}\right).$ Then, from Equation~\eqref{T3E2}, we have $\displaystyle X = \frac{u+Cv+B}{\gamma}.$ Now consider $u$ which is given by
\begin{equation}\label{Fu}
\begin{split}
  u &= \Tr_{q^{2n}/q} \left( \left(\frac{u^2+((B+Cv)^{q^2}+B+Cv)u+ C^{q^2+1}v^2+ (C^{q^2}B+CB^{q^2})v +B^{q^2+1}}{\gamma^2} \right)^{2^i} \right) \\
  &= \Tr_{q^{2}/q} \left(\frac{1}{\gamma^{2^{i+1}}} \Tr_{q^{2n}/{q^2}} \left((u^2+ C^{q^2+1}v^2+ (C^{q^2}B+CB^{q^2})v +B^{q^2+1})^{2^i} \right) \right) \\
  &= u^{2^{i+1}}\Tr_{q^{2}/q} \left(\frac{1}{\gamma^{2^{i+1}}} \Tr_{q^{2n}/{q^2}} (1)\right) +v^{2^{i+1}} \Tr_{q^{2n}/{q}} \left(\left(\frac{C}{\gamma}\right)^{2^i(q^2+1)} \right)\\
  &\quad \quad +v^{2^i} \Tr_{q^{2n}/{q}} \left(\left(\frac{C^{q^2}B+CB^{q^2}}{\gamma^2}\right)^{2^i} \right)+ \Tr_{q^{2n}/{q}} \left(\left(\frac{B}{\gamma}\right)^{2^i(q^2+1)} \right)  \\
  &= 
  \begin{cases}
   v^{2^{i+1}} \cC  + v^{2^i} \cB + B_1~&~\mbox{when}~n~\mbox{is even} \\
   u^{2^{i+1}}T_1+v^{2^{i+1}} \cC  + v^{2^i} \cB + B_1 ~&~\mbox{when}~n~\mbox{is odd},\\
  \end{cases}  
\end{split}
\end{equation}
where $\displaystyle  \cB = \Tr_{q^{2n}/q} \left( \left(\frac{BC^{q^2}+B^{q^2}C}{\gamma^2} \right)^{2^i} \right) $, $\displaystyle \cC = \Tr_{q^{2n}/q}\left( \left(\frac{C}{\gamma}\right)^{2^i(q^2+1)}\right)$ and these equalities hold because $u,v \in \F_q$, $\gamma \in \F_{q^2}^*$ and $\Tr_{q^{2n}/q} ((B+Cv)^{q^2})= \Tr_{q^{2n}/q}(B+Cv)$. Now consider $v$ which is given by 
\begin{equation}\label{Fv}
v = \Tr_{q^{2n}/q}\left(\left(\left(\frac{u+Cv+B}{\gamma}\right)^{q^4}+\left(\frac{u+Cv+B}{\gamma}\right)\right)a^{q^2}\right)=  v^{2^i}C_3+ \frac{C_1}{C},
\end{equation}
where $\displaystyle C_3= \Tr_{q^{2n}/q}\left(\left(\frac{(C^{q^4}+C)a^{q^2}}{\gamma}\right)^{2^i}\right)  $  and the second equality holds as $u, v \in \F_{q}$ and $\gamma \in \F_{q^2}^*$. Now, if $C_3=0$, then we have a unique solution $v \in \F_q$ of Equation~\eqref{Fv}. Thus, Equation~\eqref{T3E2} has a unique solution in this case. If $C_3 \neq 0$, then Equation~\eqref{Fv} can be rewritten as
$$\displaystyle v^{2^i}+ \frac{v}{C_3} +\frac{C_1}{CC_3}=0.$$
From Lemma~\ref{L1}, we know that if $\gcd(i,t)=1$, then the above equation can have at most two solutions $v \in \F_q$. Also, corresponding to each solution $v$ of Equation~\eqref{Fv}, we get a unique solution $u$ of Equation~\eqref{Fu}. Thus, for $c \in \F_{q^{2n}} \backslash \F_{q^2}$, Equation~\eqref{T3E2} has at most two solutions for each $a\in \F_{q^{2n}} \backslash \F_q$ and $b \in \F_{q^{2n}}$. This completes the proof.
\end{proof}

\section{Conclusion}\label{S5}
We have cosntructed three new classes of P$c$N permutations and two new classes of AP$c$N permutations for some specific values of $c$. Since there are only few P$c$N and AP$c$N permutations over finite fields of even characteristic, it would be an interesting problem to construct more such permutations.

\end{document}